\documentclass[11pt]{article}
\usepackage[letterpaper, margin=1in]{geometry}

\usepackage{graphicx}
\usepackage{amsmath, amssymb}
\usepackage{amsthm}
\usepackage{bm}
\usepackage{tikz}
\usepackage{caption}
\usepackage{float}
\usepackage{thm-restate}
\usepackage[style=alphabetic, doi=false,isbn=false,url=false,eprint=false,maxbibnames=99]{biblatex}

\usepackage[colorlinks=true, linkcolor=blue, citecolor=blue, urlcolor=blue]{hyperref}

\newtheorem{theorem}{Theorem}
\newtheorem{lemma}{Lemma}

\newtheorem{openquestion}{Open Question}

\DeclareMathOperator*{\disc}{\mathrm{disc}}

\title{Discrepancy And Fair Division For Non-Additive Valuations}
\author{Max Dupre la Tour, Kaito Fujii}

\bibliography{biblio.bib}

\begin{document}

\maketitle

\begin{abstract}
    We extend the notion of combinatorial discrepancy to \emph{non-additive} functions. Our main result is an upper bound of $O(\sqrt{n \log(nk)})$ on the non-additive $k$-color discrepancy when $k$ is a prime power. We demonstrate two applications of this result to problems in fair division. First, we establish a bound for a consensus halving problem, where fairness is measured by the minimum number of items that must be transferred between the two parts to eliminate envy. Second, we improve the upper bound on the total subsidy required to achieve an envy-free allocation when the number of agents is a prime power, obtaining an $O(n \sqrt{n \log n})$ bound. This constitutes the first known subquadratic guarantee in this setting.
\end{abstract}

\section{Introduction}

The celebrated Spencer's theorem in discrepancy theory~\cite{Spencer6sigma} can be naturally interpreted through the lens of a discrete consensus halving problem. Consider a set of items $ M $, and $ n $ agents with \emph{additive valuations} $ v_1, \dots, v_n : 2^{M} \rightarrow \mathbb{R} $ whose marginals lie in $ [-1,1] $. The discrete consensus-halving problem asks for a bipartition $M = S \cup S^c$ such that each agent values the two sides almost equally. Spencer's theorem guarantees a partition with per-agent imbalance $ |v_i(S) - v_i(S^c)| \leq O(\sqrt{n}) $, and this bound is asymptotically tight.

This result extends to more general settings via the theory of multi-color discrepancy~\cite{MulticolorDoerr}, corresponding to the discrete $ 1/k $-consensus splitting problem. In this setting, the objective is to partition the items into $ k \geq 2 $ subsets such that each of the $n$ agent assigns approximately equal value to each subset. One obtains $ 1/k $-consensus splits whose imbalance again scales as $ O(\sqrt{n}) $.

A natural and fundamental question is whether similar discrepancy bounds can be achieved without the assumption of additivity. Despite the intuitive appeal of this generalization, the non-additive setting appears to have received little attention. In this work, we take initial steps, providing the first results for the non-additive case.

While the theoretical question is interesting in its own right, we are further motivated by two additional considerations.

\paragraph{Consensus Halving:}

Given a bipartition of a set of items $ M = S \cup S^c $ and a valuation function $ v $, how can we determine whether the split is “fair” under that valuation?

The approach used in the definition of discrepancy considers the value gap $ |v(S) - v(S^c)| $. However, this measure is only meaningful when the valuation function has bounded marginal values—otherwise, a single outlier can dominate the gap.

For monotone valuations, a more robust fairness criterion asks: how many items must be removed from the more valuable bundle to eliminate envy? This leads to the notion of EF$c$ (envy-freeness up to $c$ items), a well-known concept in fair division generalizing the EF1 notion introduced by Lipton et al.~\cite{IntroducedEF1}. Formally, an allocation is EF$c$ if, for every agent and every other bundle, removing at most $c$ items from the latter makes the agent prefer their own bundle.

This naturally leads to an “EF$c$-based” definition of the discrepancy of a bipartition $(S, S^c)$, defined as the maximum, over a class of valuation functions, of the minimum number of items that must be removed from the richer subset to eliminate envy from the poorer subset. This concept naturally extends to $k$-color partitions and is known as ``consensus $1/k$-division up to $c$ goods,'' as introduced by Manurangsi and Suksompong~\cite{ImprovedBoundsManurangsi}, we briefly summarize some of their results below.

For additive valuations, they  established a connection between classical discrepancy and this new ``EF$c$-based'' discrepancy. In particular, known lower bounds on classical discrepancy directly yield corresponding lower bounds. With a more careful argument, they further showed that the $ O(\sqrt{n}) $ upper bound on discrepancy also implies a matching upper bound for the ``EF$c$-based'' discrepancy. In the multi-color setting, the lower bound was initially loose but was very recently improved, first by Caragiannis et al.~\cite{caragiannis2025newlowerboundmulticolor}, and then again by Manurangsi and Meka \cite{manurangsi2025tightlowerboundmulticolor}, obtaining a tight $ \Omega(\sqrt{n}) $ bound.
All of those results hold only for additive valuations, and it is also natural to explore these questions for non-additive valuations. This is left as an open question in the paper by Manurangsi and Suksompong~\cite{ImprovedBoundsManurangsi}. 

In their conclusion, they write:
\emph{``An interesting direction for future work is to consider agents with arbitrary monotonic utilities,''} 
and 
\emph{``Even in the case of prime numbers $k$, where a consensus $1/k$-division can be guaranteed for non-additive utilities~\cite{ModuloPFilos}, it is unclear whether such a division can be rounded into a discrete allocation with a loss that is bounded only in terms of $n$.''}
Our paper addresses that very challenge.

\paragraph{Envy-freeness with subsidies:}
Another motivation for studying non-additive discrepancy is to derive new results related to envy-freeness with subsidies. An allocation is envy-free when no agent would rather have another agent’s bundle than her own. For indivisible goods such an allocation may fail to exist. This impossibility disappears once a single divisible good—money—is introduced, but then a natural question arises: how large a monetary subsidy suffices to eliminate all envy?

Maskin \cite{Maskin1987}, and later Klijn \cite{subsidy1Klijn}, investigated this problem in the matching setting, where each agent can receive at most one item. Halpern and Shah \cite{HalpernShah} extended the analysis to allocations that allow agents to obtain multiple items, launching a line of follow-up work \cite{OneDollarBrustle,DichotomousBarman,TowardOptKawase}. \

In particular, Brustle et al.~\cite{OneDollarBrustle} establish a $2(n-1)^2$ upper bound on the total subsidy required for non-additive monotone valuations. This bound was recently slightly improved to $(n^2 - n - 1)/2$ by Kawase et al.\cite{TowardOptKawase}. Section~6 of the survey by Liu et al.~\cite{MixedFairDivisionSurvey} reviews this line of work and poses the following open question:

\begin{openquestion}[\cite{MixedFairDivisionSurvey}]\label{openquestion}
    Given $n$ monotone valuations with marginals in $[-1,1]$, is there an envy-free allocation with sub-quadratic total subsidy?
\end{openquestion}

\subsection{Our Contribution}

Our main result is an upper bound on the multi-color discrepancy of non-additive valuation function that hold when the number of colors is a power of a prime. 

\begin{theorem}\label{thm:mainintro}
Let $k = p^\nu$ be a power of a prime, let $M$ be a set of $m$ items and let $ v_1, \dots, v_n : 2^{M} \rightarrow \mathbb{R} $ be $n$ valuation functions with marginals in $[-1,1]$.

There exists a partition of the items in $k$ subsets $M = S_1 \cup \dots \cup S_k$ such that for all $\ell,\ell' \in [k]$ and for all $i \in [n]$: 

\[ |v_i(S_\ell) - v_i(S_{\ell'})|  = O( \sqrt{n \log nk}) \]
\end{theorem}

Note that the bound does not depend on the number of items, and that we do not assume monotonicity of the valuations. We apply this theorem in two different fair division settings.

\paragraph{Consensus Halving:} Our original goal was to extend the guarantees of Manurangsi and Skusompong\cite{ImprovedBoundsManurangsi} to non-additive valuations. While EF$c$ remains elusive in this broader setting, we establish analogous bounds for a relaxed fairness measure.

Instead of removing items, consider \emph{transferring} items between the two bundles until the lower-valued one becomes at least as valuable. The minimum number of such transfer defines the \emph{transfer imbalance}, and the worst-case value over a class of valuation functions is the \emph{transfer-discrepancy}. When valuations are not monotone, this measure is particularly natural because removing items from the higher-valued bundle, as in the EF$c$-based definition, might not reduce envy. A similar notion appears in the literature on weighted entitlements (see, e.g., WEF(1,1), as surveyed by Suksompong in~\cite{SUKSOMPONGReviewWeighted}).

As a corollary of Theorem~\ref{thm:mainintro}, we obtain the following:

\begin{restatable}{corollary}{conhal}\label{cor:consensushalving}
For any family of $n$ valuation functions, the two-color transfer-discrepancy is at most $ O\bigl(\sqrt{n \log n}\bigr) $.
\end{restatable}

Here as well, lower bounds on standard discrepancy directly imply corresponding lower bounds on transfer-discrepancy. In particular, the recent result of Manurangsi and Meka~\cite{manurangsi2025tightlowerboundmulticolor} yields an $\Omega(\sqrt{n})$ lower bound independent of the number of colors.

\paragraph{Envy-freeness with subsidies:} We show how to use partitions with small discrepancy to obtain envy-free allocations with small total subsidy. By applying Theorem~\ref{thm:mainintro}, we are able to answer Open Question~\ref{openquestion}:  
we improve the upper bound and break the quadratic barrier—\emph{even without the monotonicity assumption}—when the number of agents is a power of a prime.

\begin{restatable}{corollary}{subsidy}\label{cor:subsidyintro}
For arbitrary valuations with marginals in the interval $[-1,1]$, if the number of agents $n = p^\nu$ is a power of a prime, there exists an envy-free allocation requiring a total subsidy of $O\left(n\sqrt{n\log n}\right)$.
\end{restatable}

Table \ref{tab:knownbounds} compares our result with the known upper bounds on the total subsidy required for various classes of valuations. 

\begin{table}[H]
\centering
\begin{tabular}{|c|c|c|}
\hline
\textbf{Valuation class} & \textbf{Upper bound on  subsidy} & \textbf{Reference} \\
\hline
Monotone + additive & $m(n-1)$ & \cite{HalpernShah} \\
\hline
Monotone + additive & $n-1$ & \cite{OneDollarBrustle} \\
\hline
Monotone & $2(n-1)^2$ & \cite{OneDollarBrustle} \\
\hline
Monotone + dichotomous & $n-1$ & \cite{DichotomousBarman} \\
\hline
Monotone & $(n^{2}-n-1)/2$ & \cite{TowardOptKawase} \\
\hline
Doubly monotone & $n(n-1)/2$ & \cite{TowardOptKawase} \\
\hline
Arbitrary ($n$ is a prime power) & $O\bigl(n\sqrt{n\log n}\bigr)$ & Corollary~\ref{cor:subsidyintro} \\
\hline
\end{tabular}
\caption{Upper bounds on the monetary subsidy required for envy-freeness}
\caption*{Note: $m$ is the number of items and $n$ the number of agents. Every marginal item value lies in the interval $[-1,1]$. A valuation is monotone if all marginal values are non-negative, dichotomous if each marginal is either $0$ or $1$, and doubly monotone if every item is classified as either a good (all its marginals are non-negative) or a chore (all its marginals are non-positive).\\}
\label{tab:knownbounds}
\end{table}
\vspace{-1cm}

\subsection{Overview of Our Approach}

\paragraph{Proof of Theorem~\ref{thm:mainintro}:}
The standard proof of the two-color discrepancy bound for additive valuation functions (see the Discrepancy chapter of \cite{ProbabilisticMethod}) works as follows.  
When $m>n$, a dimension-reduction lemma from \cite{DimensionReductionBARANY}---using linear algebra---yields a fractional allocation in which every set has identical value under all functions and at most $n$ coordinates are non-integral. A naïve rounding of this solution, coupled with Hoeffding's inequality and union bounds, gives an $ O(\sqrt{n\log n})$ bound. .

With more refined rounding, the union bound can be avoided and the guarantee improves to $O(\sqrt{n})$. Reaching this tighter bound reduces to the special case $m=n$. Several proofs have been found (\cite{Spencer6sigma,DiscrepancyLovett,DiscrepancyRothvoss,DiscrepancyWeightLevy,DiscrepancyPesenti}), but rely—at least implicitly—on the partial-coloring method, which recursively colors a constant fraction of items. This technique crucially exploits additivity, because the total discrepancy is built up from independent contributions at each recursive stage. In the non-additive setting, it is unclear whether the extra $\sqrt{\log n}$ factor imposed by the union bound can be eliminated.

Our approach follows essentially the same steps. We define the value of fractional allocations using the multilinear extension~\cite{MultilinearExtension}, which extends a set function to the continuous domain $[0,1]^m$ by taking the expected value of the function when each item is included independently with a given probability.

Our goal is to obtain a fractional allocation in which all bundles have equal value according to the $n$ multilinear extensions of the valuation functions, and each bundle contains only a few fractional items. Linear-algebraic tools no longer suffice, so we turn instead to continuous necklace-splitting results. Alon’s continuous necklace-splitting theorem \cite{NecklacesAlon} states that a necklace $[0,1]$ endowed with $n$ additive continuous measures $\mu_1,\dots,\mu_n$ can be cut $n(k-1)$ times and partitioned into $k$ bundles of equal value for every measure.

Later work (e.g., \cite{ConsensusHalvingSimmons} for $k=2$ and \cite{ModuloPFilos} for larger prime $k$) observed that additivity is needed only when $k$ is composite. For prime $k$, the proof relies on a Borsuk--Ulam-type theorem and therefore requires only continuity. When $k$ is composite, the necklace is recursively cut, and it is only this recursive step that relies on additivity.

For the case $k = 2$, we can apply results from necklace splitting as follows:  
we model each item as a short interval, align all items along the interval $[0,1]$,  
and invoke the non-additive two-bundle necklace-splitting theorem of~\cite{ConsensusHalvingSimmons}.  
The resulting allocation requires at most $n$ cuts, and thus introduces at most $n$ fractional items.  
A similar approach is used by Goldberg et al.~\cite{ConsensusHalvingGoldberg}.

For $k > 2$, to ensure that each bundle contains at most $O(n)$ fractional items, it suffices that it contains only $O(n)$ pieces (intervals) after the $n(k - 1)$ cuts, since each bundle contains at most one fractional item per endpoint of the intervals it receives.

This was established in \cite{NecklaceConstraintsJojic}, under the assumptions that $ k $ is a power of a prime and that the measures are additive. However, as before, their proof uses only continuity (to invoke the Borsuk-Ulam type theorem), not additivity. Thus, their result extends to our setting. We therefore obtain a fractional allocation in which every bundle is equal under the $n$ multilinear extensions and contains $O(n)$ non-integral items.

Finally, we perform independent randomized rounding. Using McDiarmid’s inequality \cite{McDiarmid_1989}—a generalization of the Hoeffding's inequality— and union bound, we obtain an $O(\sqrt{n\log nk})$ discrepancy guarantee. The multilinear extension is essential: it expresses each bundle’s fractional value as the expectation of the random assignment, making McDiarmid’s bound directly applicable.

\paragraph{Consensus Halving:}
We now outline the proof of Corollary~\ref{cor:consensushalving}. Given a collection $\mathcal{V}$ of $n$ valuation functions, we define a transformed collection $\mathcal{V}'$ as follows. For each $v \in \mathcal{V}$, we define a corresponding function $v' \in \mathcal{V}'$ by  defining $v'(S)$ as the minimum number of items that must be exchanged between $S$ and $S^c$ to reverse the inequality between $v(S)$ and $v(S^c)$. The key observation is that each $v'$ has marginals in $\{-1, 0, 1\}$. This property allows us to directly apply Theorem~\ref{thm:mainintro} with two colors to the transformed family $\mathcal{V}'$, yielding Corollary~\ref{cor:consensushalving}.

\paragraph{Envy-freeness with subsidies:}

To derive Corollary~\ref{cor:subsidyintro} from Theorem~\ref{thm:mainintro}, we use a lemma by Halpern and Shah~\cite{HalpernShah}, which characterizes the allocations that can be made envy-free with payments. They prove that these are precisely the allocations that already maximize total welfare across all possible reassignments of the bundles. A convenient way to understand this is via the envy graph: a complete directed graph whose vertices represent the agents and whose edge weights indicate how much better or worse off an agent would be if they swapped bundles with another. If this graph contains no cycle with positive total weight, then suitable payments exist that eliminate all envy. Each agent’s payment can be set to the weight of the maximum-weight path starting from that agent.

Our strategy builds on this insight. We first break the items into bundles using a partition with low discrepancy. We then assign these bundles so that overall welfare is maximized. Because the resulting envy graph has no positive cycle, any path’s total weight is bounded by the single back-edge that would close a cycle. That weight of that back-edge is bounded by the discrepancy of the original partition. Consequently, the total payments we need are at most $(n-1)$ times that discrepancy.

\subsection{Concurrent work}

Independently and in parallel, Hollender, Manurangsi, Meka, and Suksompong~\cite{hollender2025} have studied the problem of extending discrepancy and envy-freeness for groups to non-additive valuations. Their work builds on the same high-level ideas as ours, namely the use of multilinear relaxations combined with necklace splitting to obtain a fair fractional allocation, followed by rounding via McDiarmid’s inequality. While the overall approaches are very similar, there are noteworthy differences. 

In the two-color case, they apply the two-color bound in a more direct and elegant way, which yields guarantees for envy-freeness up to $c$ goods under monotone valuations. Our approach instead leads to a weaker variant, in which items from the more valuable side are transferred rather than removed; this has the advantage of also applying in the non-monotone setting. 

In the multi-color case, we rely on a necklace splitting theorem of Jojić et al.~\cite{NecklaceConstraintsJojic}. This allows us to handle prime powers (rather than just primes) while ensuring that each agent receives $O(n)$ intervals, resulting in a fractional allocation with only $O(n)$ fractional items per bundle. Consequently, our analysis improves their bound from $O(\sqrt{nk \log(nk)})$ to $O(\sqrt{n \log(nk)})$. 

It is this improvement, used in the diagonal case $k=n$, that allows us to obtain improved subquadratic guarantees in the “fairness with subsidy” framework.

Both works were developed independently, and the close parallelism highlights the robustness of the underlying ideas.

\subsection{Further Related Work}

Necklace splitting and consensus halving have recently attracted considerable attention, particularly due to hardness results that establish a deep connection with the complexity class~PPA~\cite{furtherFilosPPA,furtherFilosPPA3,furtherDeligkas3,furtherDeligkas1,ModuloPFilos,Deligkas2,furtherFilosPPA4}.

Recent papers explore other directions related to envy-freeness with subsidies: addressing settings with weighted entitlements~\cite{aziz2024weightedenvyfreeallocationsubsidy,elmalem2025saidmoneywontsolve}; achieving envy-freeness with small subsidies and good Nash welfare under additive valuations~\cite{furtherTwoBirdsNarayan}; designing truthful mechanisms for envy-free allocations with small subsidies for matroid rank valuations~\cite{furtherGOKO}; analyzing the hardness of computing the minimum required subsidy~\cite{furtherCaragiannis}; and considering alternative fairness notions (proportionality or equitability) with subsidies~\cite{furtherWuZhou1,furtherWuZhou2,furtherAziz}.

\section{Preliminaries}

We introduce the definitions needed to state and prove our main theorem. Definitions specific to the applications will be provided in the corresponding sections.

Let $N = [n]$ be a set of $n$ agents, and let $M$ be a set of $m$ items. Each agent $i \in N$ has a valuation function $v_i : 2^M \to \mathbb{R}$. Let $\mathcal{V} = \{v_1, \dots, v_n\}$ denote the set of valuation functions.

A valuation function $v_i$ has \emph{marginals bounded by} $L$ if, for all $S \subseteq M$ and all $j \in M$, we have
\[
|v_i(S \cup \{j\}) - v_i(S)| \leq L.
\]

To move from discrete to continuous settings, we define the \emph{multilinear extension} of a valuation function~\cite{MultilinearExtension}. For a fractional vector $x \in [0,1]^M$, where $x_j$ represents the marginal probability of including item $j$, the multilinear extension $F_i : [0,1]^M \to \mathbb{R}$ of $v_i$ is defined as
\[
F_i(x) = \mathbb{E}_{S \sim x}[v_i(S)],
\]
where each item $j \in M$ is included in the random subset $S$ independently with probability $x_j$, i.e., $\Pr[j \in S] = x_j$.

A \emph{$k$-coloring} of $M$ is a function $\chi : M \to [k]$ that assigns each item a color from the set $\{1, \dots, k\}$, inducing the partition $M = \chi^{-1}(1) \cup \dots \cup \chi^{-1}(k)$.

A \emph{fractional $k$-coloring} generalizes this by allowing probabilistic assignments. It is a function $\chi : M \to \Delta_k$, where $\Delta_k$ is the $(k-1)$-dimensional simplex
\[
\Delta_k = \left\{(p_1, \dots, p_k) \in [0,1]^k \;\middle|\; \sum_{\ell=1}^k p_\ell = 1 \right\}.
\]
Each item $j \in M$ is thus assigned a probability distribution over the $k$ colors.

\subsection{Multi-color Discrepancy For Valuation Functions}

Let $k \geq 2$ be a number of colors, and let $\chi$ be a $k$-coloring of $M$. The \emph{discrepancy} of $\mathcal{V}$ with respect to $\chi$ is defined as
\[
\disc(\mathcal{V}, k, \chi) = \max_{\ell,\ell' \in [k]} \max_{1 \leq i \leq n} \left| v_i(\chi^{-1}(\ell)) - v_i(\chi^{-1}(\ell')) \right|.
\]
The $k$-color discrepancy of $\mathcal{V}$ is then defined as
\[
\disc(\mathcal{V}, k) = \min_{\chi : M \rightarrow [k]} \disc(\mathcal{V}, k, \chi).
\]

This definition naturally generalizes the classical notion of discrepancy. Given a matrix $ A \in \mathbb{R}^{n \times m} $, each row of $ A $ can be interpreted as an additive valuation function over the item set $ M $. The discrepancy of the resulting collection of valuation functions coincides—up to a factor of 2—with the matrix-based definition of multi-color discrepancy of Doerr and Srivastav \cite{MulticolorDoerr}.

\subsection{McDiarmid’s inequality}
We state here the uniform version of McDiarmid’s inequality, which will be used later to control the error of our randomized rounding schemes.

\begin{lemma}[McDiarmid's inequality\cite{McDiarmid_1989}]\label{lem:McDiarmid}
    
Let \( X = (X_1, \dots, X_t) \) be a vector of independent random variables, where \( X_i \) takes values in \( \mathcal{X}_i \) for \( i = 1, \dots, t \), and let  
$
f : \mathcal{X}_1 \times \dots \times \mathcal{X}_t \to \mathbb{R}
$
satisfy the uniform bounded–differences condition:
\[
\bigl| f(x) - f(x') \bigr| \le L
\quad \text{whenever } x, x' \text{ differ in at most one coordinate}.
\]

\noindent Then, for every $ a > 0 $,
\[
\Pr\!\left[\,\left|f(X) - \mathbb{E}f(X)\right| \ge a\,\right]
\;\le\;
2 \exp\!\left(-\frac{2a^2}{tL^2}\right).
\]
\end{lemma}

In words: if changing any single input can affect the value of $ f $ by at most $ L $, then $ f(X) $ concentrates around its expectation with sub-Gaussian tails governed by variance proxy $ tL^2/2 $.

\section{Non-additive Discrepancy}

In this section, we present and prove our main result, which is exactly Theorem~\ref{thm:mainintro}, restated here using our updated notation.

\begin{theorem}\label{thm:main}
Let $ k = p^\nu $ be a prime power, and let $ \mathcal{V} = \{v_1, \dots, v_n\} $ be a collection of valuation functions with marginals bounded by 1\footnote{If the marginals are bounded by a constant $ L $, the result holds with an additional multiplicative factor of $ L $, via rescaling.}. Then the multi-color discrepancy satisfies
\[
\disc(\mathcal{V}, k) = O\left( \sqrt{n \log nk} \right).
\]
\end{theorem}

To prove this theorem, we begin by showing that it suffices to construct a fractional coloring that induces a fractional partition in which all parts have equal values under the multilinear extensions of the valuations in $\mathcal{V}$, and in which each part contains only a few fractionally assigned items. This is formalized in the next lemma.

\begin{lemma}\label{lem:rounding}
Let $ t $ be a parameter. Suppose there exists a fractional coloring $ \chi : M \to \Delta_k $ such that for each $ \ell \in [k] $, the vector $ \chi_{\ell} := \left( \chi(j)_{\ell} \right)_{j \in M} \in [0,1]^M $, representing the fraction of each item assigned to color $ \ell $, satisfies:

\begin{enumerate}
    \item For every agent $ i \in [n] $ and all $ \ell, \ell' \in [k] $, we have
    \[
    F_i(\chi_{\ell}) = F_i(\chi_{\ell'}).
    \]
    \item For each $ \ell \in [k] $, the vector $ \chi_{\ell} $ has at most $ t $ non-integral coordinates.
\end{enumerate}
Then,
\[
\disc(\mathcal{V}, k) = O\left(\sqrt{t\log nk}\right).
\]
\end{lemma}

\begin{proof}
We round the fractional allocation using a natural randomized scheme: for each item $ j \in [m] $, independently assign it to color $ \ell $ with probability $ \chi_\ell(j) $. Let $ S_\ell $ denote the resulting random bundle of color $ \ell $.

For each agent $ i $, define
\[
\mu_i := F_i(\chi_1) = \cdots = F_i(\chi_k) = \mathbb{E}\bigl[v_i(S_1)\bigr] = \cdots = \mathbb{E}\bigl[v_i(S_k)\bigr],
\]
where the second sequence of equalities follows from the definition of the multilinear extension.

\medskip

Fix an agent $ i $ and a color $ \ell $, and let
\[
J_\ell := \{ j \in [m] \mid \chi_\ell(j) \in (0,1) \}, \qquad |J_\ell| \le t.
\]
Define $ T_\ell := \{ j \in [m] \mid \chi_\ell(j) = 1 \} $, the set of items always assigned to color $ \ell $.

Enumerate $ J_\ell = \{ j_1, \dots, j_{t'} \} $ with $t' \leq t$, and for each $ s \in [t'] $, define the indicator variable
\[
X_s :=
\begin{cases}
1 & \text{if } j_s \in S_\ell, \\
0 & \text{otherwise}.
\end{cases}
\]

\smallskip

We define a transformed valuation function to apply McDiarmid’s inequality (Lemma~\ref{lem:McDiarmid}):
\[
v_i'\colon \{0,1\}^{t'} \to \mathbb{R}, \qquad
v_i'(x_1,\dots,x_{t'}) := v_i\left(T_\ell \cup \{ j_s \in J_\ell \mid x_s = 1 \} \right).
\]
Then the random vector $ X := (X_1,\dots,X_t) $ consists of independent bits, and
\[
v_i(S_\ell) = v_i'(X), \qquad
\mu_{i}  = \mathbb{E}[v_i'(X)].
\]

\smallskip

Flipping a single coordinate $ x_s $ changes the value of the bundle by at most $1$, since marginals of $v_i$ lie in $ [-1,1] $. Applying Lemma~\ref{lem:McDiarmid} yields:
\[
\Pr\left[\,|v_i(S_\ell) - \mu_{i}| \ge a\,\right]  = \Pr\left[\,|v'_i(X) - \mu_{i}| \ge a\,\right]\le 2\exp\left(-\frac{2a^2}{t'}\right) \leq 2\exp\left(-\frac{2a^2}{t}\right) .
\]

\medskip

We now choose $ a $ large enough to apply a union bound over all $ nk $ agent-color pairs. Let $a :=\sqrt{c\,t \log(nk)}\quad \text{for some constant } c$. Then,
\[
\Pr\left[\,|v_i(S_\ell) - \mu_{i}| \ge a\,\right] \le 2(nk)^{-2c}.
\]
Choosing $c$ such that $ 1 - 2(nk)^{-2c + 1}> 0  $, a union bound over all $ nk $ pairs $ (i,\ell) $ shows that with probability at least $ 1 - 2(nk)^{-2c + 1} $, we have for all $ i \in [n] $, $ \ell \in [k] $:
\[
|v_i(S_\ell) - \mu_{i}| \le O\left(\sqrt{t\log(nk)}\right).
\]

Combining this with the triangle inequality we get:

\[
|v_i(S_\ell) - v_i(S_{\ell'})| \leq |v_i(S_\ell) - \mu_i| + |v_i(S_{\ell'}) - \mu_i|  =  O\left(\sqrt{t\log(nk)}\right)\quad \text{for all } i \in [n],\; \ell, \ell' \in [k].
\]

This happens with positive probability, in particular there exists a coloring with discrepancy bounded by \( O\left(\sqrt{t \log(nk)}\right) \).

\end{proof}

Our next task is to exhibit a fractional allocation that satisfies Lemma~\ref{lem:rounding} with
$t = O(n)$.  To do so we invoke a theorem of Jojić et al.~\cite{NecklaceConstraintsJojic} which strengthen Alon's classical result on necklace splitting~\cite{NecklacesAlon}.

\begin{theorem}[\cite{NecklaceConstraintsJojic}]\label{thm:Jojic}
Let $k = p^{\nu}$ be a power of a prime.
For any $n$ continuous set functions
$\mu_1,\dots,\mu_n$ defined on disjoint unions of intervals in $[0,1]$—where inserting or deleting a degenerate (length‑zero) interval leaves the function value unchanged—there is a partition of $[0,1]$ into $k$ bundles using at most $n(k-1)$ cuts such that:
\begin{itemize}
    \item Each bundle is equal according to all $\mu_i$, and
    \item Each bundle contains fewer than $\bigl\lfloor n(k-1)/k \bigr\rfloor + 1$ intervals.
\end{itemize}
\end{theorem}

The published statement in~\cite{NecklaceConstraintsJojic} assumes that the functions $\mu_i$ are additive probability measures. However, additivity is never used in the proof—the argument relies solely on continuity to invoke a Borsuk--Ulam-type theorem. Thus, Theorem~\ref{thm:Jojic} holds in the broader generality stated above.

Similar observations have been made in related settings: Simmons and Su~\cite{ConsensusHalvingSimmons} for the case $k = 2$, and Filos-Ratsikas et al.~\cite{ModuloPFilos} for prime $k$, both noted that Alon's result using at most $n(k-1)$ cuts does not require additivity. Theorem~\ref{thm:Jojic} goes slightly further, as it applies to prime powers $k$ and includes the additional constraint that each bundle contains only a small number of intervals. With Theorem~\ref{thm:Jojic} in hand we can prove Theorem \ref{thm:main}.

\begin{proof}

Model each (indivisible) item as a tiny disjoint interval in $[0,1]$.
For any measurable subset $S\subseteq[0,1]$ that is a disjoint union of intervals, define
$\mu_i(S) = F_i(\mathbf{x}(S))$,
where $F_i$ is the multilinear extension of player $i$’s valuation and $\mathbf{x}(S)$ is the corresponding incidence vector. 
The incidence vector $\mathbf{x}(S) \in [0,1]^m$ encodes the fractional inclusion of items in $S$: the $j$th coordinate $\mathbf{x}(S)_j$ equals the measure of the intersection between item~$j$’s interval and $S$.
Since each $F_i$ is continuous, the hypotheses of Theorem~\ref{thm:Jojic} are satisfied.

Applying the theorem yields a partition of $[0,1]$ into $k$ bundles using $n(k-1)$ cuts, with every bundle consisting of at most
$\lfloor n(k-1)/k \rfloor + 1 \le n$ intervals.
Along each cut, at most one item is sliced, so each bundle contains at most two fractional items per interval in the bundle, i.e. at most $2n$ fractional items in total.
Moreover, by construction all bundles are exactly equal according to every $F_i$.
Consequently we obtain a fractional allocation that meets Lemma~\ref{lem:rounding} with
$t = 2n$, completing the proof of Theorem~\ref{thm:main}.
\end{proof}
Goldberg et al.~\cite{ConsensusHalvingGoldberg} used an analogous idea for the case $k=2$, using Simmons and Su~\cite{ConsensusHalvingSimmons} result to obtain a fractional allocation with at most $n$ fractional items per bundle.

Extending Theorem~\ref{thm:Jojic} beyond prime powers would directly generalize Theorem~\ref{thm:main}. However it remains open even without the extra requirement that each bundle contain at most $n$ pieces. For composite $k$, the standard proof for necklace splitting uses a recursive cutting procedure that crucially depends on additivity; dropping that assumption breaks the recursion and no alternative argument is known.

\section{Fair Consensus Halving}\label{sec:consensus-halving}

In this section we extend the framework of Manurangsi and Suksompong~\cite{ImprovedBoundsManurangsi} to the non‑additive setting in the consensus halving case ($k=2$).  
Our goal is to split the set $M$ of indivisible items into two blocks $S$ and $S^{c}$ in such a way that every agent regards the division as approximately fair.  
The quality of a bipartition is measured by the minimum number of items that must cross the cut in order to make the ``poorer'' side richer for that agent.  
We call this quantity the \emph{transfer‑imbalance}.  Unlike discrepancy it can be kept small without assuming bounded marginals.
For monotone valuations, Manurangsi and Suksompong measured fairness by how many items have to be \emph{removed} from the richer bundle to break the tie (``consensus $1/k$ up to $c$ items'' in their paper). Their definition is closely related to the notion of \emph{envy-freeness up to $c$ items (EF$c$)}, extensively studied in the fair division literature.

Our transfer notion is not as standard, but it appears as well in fair division, for example in the literature on weighted entitlement (See WEF(1,1) in the review by Suksompong \cite{SUKSOMPONGReviewWeighted}).
Replacing removal by two‑sided transfers is slightly weaker, but it makes the definition meaningful even for non‑monotone valuations.   

\paragraph{Transfer-Imbalance:}For a valuation $v$ and two disjoint sets $S,T\subseteq M$ the transfer imbalance between $A$ and $B$ is defined as:
\[
T_v(A, B) := 
\begin{cases}
\min\left\{\,|S| + |S'| \;\middle|\; 
\begin{array}{l}
S \subseteq A,\; S' \subseteq B, \\[2pt]
v\big((A \setminus S) \cup S'\big)
\leq
v\big((B \setminus S') \cup S\big)
\end{array}
\right\}
& \text{if } v(A) \ge v(B), \\[12pt]
T_v(B, A)
& \text{if } v(B) > v(A).
\end{cases}
\]

Intuitively, $T_{v}(A,B)$ is the smallest number of items that must be transferred from one side to the other to reverse the inequality between $v(A)$ and $v(B)$.

\paragraph{Transfer-Discrepancy:} Let $k \geq 2$ be a number of colors, and let $\chi$ be a $k$-coloring of $M$. The \emph{transfer-discrepancy} of $\mathcal{V}$ with respect to $\chi$ is defined as
\[
\textstyle\disc^T(\mathcal{V}, k, \chi) = \max_{\ell,\ell' \in [k]} \max_{1 \leq i \leq n} T_{v_i}(\chi^{-1}(\ell),\chi^{-1}(\ell')).
\]
The $k$-color transfer-discrepancy of $\mathcal{V}$ is then defined as
\[
\textstyle\disc^T(\mathcal{V}, k) = \min_{\chi : M \rightarrow [k]} \disc(\mathcal{V}, k, \chi).
\]

Whenever all marginal values are at most $1$, removal and transfer differ by at most a factor of~$2$, so existing lower bounds immediately carry over.

\begin{restatable}{proposition}{proplower}\label{prop:T-lb}
For every $k\ge 2$, there exists a collection of $n$ (additive) valuation functions with $k$-color transfer-discrepancy at least $\Omega(\sqrt n)$
\end{restatable}

\begin{proof}

The lower bound for (additive) multi-color discrepancy have been recently improved to $\Omega(\sqrt n)$ by Manurangsi and Meka ~\cite{manurangsi2025tightlowerboundmulticolor}, and now match the additive upper bound. 
Because the valuations in their construction have marginals bounded by $1$, moving a single item changes the additive imbalance by at most~$1$.  Hence for every two color classes $A,B$ and every agent $v$,
\[ 
  T_{v}(A,B)\;\ge\;\frac{\lvert v(A)-v(B)\rvert}{2},
\]
so the same family witnesses $\disc^{T}(\mathcal{V},k)=\Omega(\sqrt{n})$.
\end{proof}

\subsection{Upper bound for consensus halving}

We now show that our Theorem~\ref{thm:main} implies a $O(\sqrt{n \log n})$ upper bound in the two color case.

\conhal*

The proof relies on the following Lipschitz property of the transfer‑imbalance.

\begin{lemma}\label{lem:Lipschitz}
For every valuation $v$ and every $S\subseteq M,\;x\in M$,
\[
  \bigl|
    T_{v}(S,S^{c}) \;-\; T_{v}(S\triangle\{x\},\,(S\triangle\{x\})^{c})
  \bigr|
  \;\le\; 1.
\]
\end{lemma}

\begin{proof}
Assume w.l.o.g.\ that $x\in S$ (the other case is symmetric).  
Let $t=T_{v}(S,S^{c})$ witnessed by a transfer set $R\subseteq S$ of size $\lvert R\rvert=t$.  
If $x\notin R$ then $R$ also witnesses an imbalance between $S\triangle\{x\}$ and its complement, so $T_{v}(S\triangle\{x\},(S\triangle\{x\})^{c})\le t$.  
If $x\in R$ then $R\setminus\{x\}$ witnesses an imbalance after the flip, of size at most $t-1$.  
In either case the new transfer‑imbalance differs from $t$ by at most~$1$.
\end{proof}

Lemma \ref{lem:Lipschitz} shows that $S\mapsto T_{v}(S,S^{c})$ is 1‑Lipschitz with respect to Hamming distance. We are now ready to prove Corollary~\ref{cor:consensushalving}

\begin{proof}

Fix an agent $i \in N$.  Define
\[
  v'_{i}(S)
  \;=\;
  \begin{cases}
    +T_{v_{i}}(S,S^{c}) & \text{if } v_{i}(S)\;\ge\;v_{i}(S^{c}),\\[4pt]
    -T_{v_{i}}(S,S^{c}) & \text{otherwise.}
  \end{cases}
\]

It is anti-symmetric $v'_{i}(S)=-v'_{i}(S^{c})$ and by Lemma~\ref{lem:Lipschitz},$\bigl|v'_{i}(S\triangle\{x\})-v'_{i}(S)\bigr|\le 1$ for every item $x$. Thus, every $v'_{i}$ satisfies the assumptions of Theorem~\ref{thm:main} with \emph{marginals bounded by~$1$} and there exists a partition $M = S \cup S^c$ such that \[\max_{i\in[n]} \bigl|v'_{i}(S) - v'_i(S^c)\bigr| =\max_{i\in[n]} \bigl|2v'_{i}(S) \bigr|  = O(\sqrt{n \log n}).\]
This concludes the proof of Corollary~\ref{cor:consensushalving}. 

\end{proof}

\section{Envy-freeness With Subsidies}
We begin this section by outlining the setting introduced by Halpern and Shah \cite{HalpernShah}.

An \emph{allocation} is an ordered partition $\bm{A} = (A_1, \dots, A_n)$ of the item set $M$. Let $\bm{p} = (p_1, \dots, p_n) \in \mathbb{R}^n_{\geq 0}$ denote a \emph{payment vector}. The pair $(\bm{A}, \bm{p})$ is said to be \emph{envy-free} if, for every pair of agents $i, j \in N$, we have $v_i(A_i) + p_i \geq v_i(A_j) + p_j$. The goal is to find an envy-free allocation that minimizes the total payment $\sum_{i=1}^n p_i$.

An allocation $\bm{A}$ is called \emph{envy-freeable} if there exists a payment vector $\bm{p}$ such that $(\bm{A}, \bm{p})$ is envy-free.
The \emph{envy graph} $G_{\bm{A}}$ associated with an allocation $\bm{A} = (A_1, \dots, A_n)$ is the complete directed graph on the set of agents $N$, where the weight of arc $(i, j)$ is given by the envy agent $i$ feels toward agent $j$: $w_{\bm{A}}(i, j) = v_i(A_j) - v_i(A_i)$.
The weight of a path (or cycle) in $G_{\bm{A}}$ is defined as the sum of the weights of its arcs.

The following lemma—central to all subsequent work in this setting—characterizes envy-freeable allocations and determines the minimal payment vector required.

\begin{lemma}[Halpern and Shah \cite{HalpernShah}]\label{lem:HalpernShah}
The following statements are equivalent:
\begin{enumerate}
    \item The allocation $\bm{A} = (A_1, \dots, A_n)$ is envy-freeable.
    \item $\bm{A}$ maximizes welfare over all reassigned bundles: for every permutation $\sigma : N \rightarrow N$,
    \[
    \sum_{i=1}^n v_i(A_i) \geq \sum_{i=1}^n v_i(A_{\sigma(i)}).
    \]
    \item The envy graph $G_{\bm{A}}$ contains no positive-weight cycle.
\end{enumerate}

Moreover, if these conditions hold, the \emph{minimal payment vector} $\bm{p}$ that makes $(\bm{A}, \bm{p})$ envy-free is defined by setting $p_i$ to be the maximum weight of any path starting at vertex $i$ in $G_{\bm{A}}$. The absence of positive cycles ensures that this is well-defined.
\end{lemma}

Our next lemma provides an upper bound on the total payment required, expressed in terms of the discrepancy of the valuations.

\begin{lemma}\label{lem:subsidydiscrepancy}
There exists an allocation $\bm{A} = (A_1, \dots, A_n)$ and a payment vector $\bm{p} = (p_1, \dots, p_n)$ such that the pair $(\bm{A}, \bm{p})$ is envy-free and, for all $i \in N$, it holds that $p_i \leq \disc(\mathcal{V}, n)$. In particular, the total payment satisfies:
\[
\sum_{i=1}^n p_i \leq (n-1)\disc(\mathcal{V}, n).
\]
\end{lemma}

\begin{proof}
Let $ \chi: M \rightarrow [n] $ be an $ n $-coloring of $ M $ that minimizes $ \disc(\mathcal{V}, n, \chi) $. This coloring induces a collection of bundles, where for each $ i \in [n] $, we define $ A_i = \chi^{-1}(i) $. We then consider all possible reassignments of the bundles to the agents and reorder the indices so that the resulting allocation $ \bm{A} = (A_1, \dots, A_n) $ maximizes the total welfare among all such permutations. By Lemma~\ref{lem:HalpernShah}, the allocation $ \bm{A} $ is envy-freeable. If we define $ p_i $ as the maximum weight of any path starting at vertex $ i $ in the graph $ G_{\bm{A}} $, then the payment vector $ \bm{p} = (p_1, \dots, p_n) $ ensures that the pair $ (\bm{A}, \bm{p}) $ is envy-free. It remains to bound each $ p_i $.

Let $i$ be an agent, and let $j$ be the last vertex in a maximum-weight path starting at vertex $i$. By Lemma~\ref{lem:HalpernShah}, the graph $G_{\bm{A}}$ contains no positive-weight cycles. Therefore, $p_i + w_{\bm{A}}(j, i) \leq 0$, which implies that $p_i \leq |v_j(A_i) - v_j(A_j)| \leq \disc(\mathcal{V}, n)$. Moreover at least one of the $p_i$ must be zero. Hence, the total payment satisfies $\sum_{i=1}^n p_i \leq (n - 1) \cdot \disc(\mathcal{V}, n)$.

\end{proof}

Corollary~\ref{cor:subsidyintro} directly follows from Lemma~\ref{lem:subsidydiscrepancy} and Theorem~\ref{thm:main}.

\subsidy*

\section{Conclusion}

We introduced a natural generalization of combinatorial discrepancy to non-additive functions and proved an $O(\sqrt{n \log(nk)})$ upper bound for the $k$-color case when $k$ is a prime power. Our work opens several directions for future research.

A natural question is whether our bound can be extended to all values of $k$, beyond prime powers. Additionally, there remains a small $\sqrt{\log(nk)}$ gap between the upper and lower bounds. Can this gap be closed? Even in the case $k = 2$, it is unclear whether the extra $\sqrt{\log n}$ factor is inherent, or if it is possible to ``beat the union bound'', as in the additive case.

Finally, exploring further implications of non-additive discrepancy is a promising direction. While we focused on applications to fair division, we believe that the generality of our bound suggests potential applications in other areas as well.

\printbibliography

\end{document}